\documentclass[12pt]{article}
\usepackage[margin=1.3in]{geometry}
\usepackage[centertags]{amsmath}
\usepackage{amsfonts,amsthm,amssymb}
\usepackage{amssymb}
\usepackage{amsmath}
\usepackage{graphicx}
\usepackage{bbm}
\usepackage{color}
\usepackage{mathtools}

\theoremstyle{definition}
\newtheorem{theorem}{Theorem}

\newtheorem{definition}{Definition}
\newtheorem{lemma}{Lemma}
\newtheorem{proposition}{Proposition}
\newtheorem{remark}{Remark}
\newtheorem{example}{Example}

\def\E{\mathbb{E}}

\DeclareMathOperator*{\argmax}{arg\,max}

\begin{document}

\title{Small-Support Approximate Correlated Equilibria}
\author{Yakov Babichenko\footnote{Center for the Mathematics of Information, California Institute of Technology. {\tt babich@caltech.edu}} \and Siddharth Barman\footnote{Center for the Mathematics of Information, California Institute of Technology. {\tt barman@caltech.edu}} \and Ron Peretz\footnote{Department of Mathematics, London School of Economics. {\tt ronprtz@gmail.com}}}

\date{}
\maketitle

\begin{abstract}
We prove the existence of approximate correlated equilibrium of support size polylogarithmic in the number of players and the number of actions per player. In particular, using the probabilistic method, we show that there exists a multiset of polylogarithmic size such that the uniform distribution over this multiset forms an approximate correlated equilibrium. Along similar lines, we establish the existence of approximate coarse correlated equilibrium with logarithmic support. 
  
We complement these results by considering the computational complexity of determining small-support approximate equilibria. We show that random  sampling can be used to efficiently determine an approximate coarse correlated equilibrium with logarithmic support. But, such a tight result does not hold for correlated equilibrium, i.e., sampling might generate an approximate correlated equilibrium of support size $\Omega(m)$ where $m$ is the number of actions per player. Finally, we show that finding an \emph{exact} correlated equilibrium with smallest possible support is \rm{NP}-hard under Cook reductions, even in the case of two-player zero-sum games.

\end{abstract}

\section{Introduction}
Equilibria are central solution concepts in the theory of strategic games. Arguably the most prominent examples of such notions of rationality are Nash equilibrium~\cite{Nash}, correlated equilibrium~\cite{aumann74,aumann87}, and coarse correlated equilibrium~\cite{hannan}. At a high level, these concepts denote distributions over players' strategy profiles where no player can benefit, in expectation, by unilateral deviation. Equilibria are used to model the outcomes of interaction between strategic human players, and between organizations run by human agents.  Hence, if a solution concept is too complicated (say, on account of the fact that it requires randomization over a large set of strategy profiles) then its applicability is debatable, simply because it is hard to imagine that human players would adopt highly intricate strategies. Such concerns have been raised in the context of bounded rationality, see, e.g.,~\cite{simon} and~\cite{rubinstein}. Therefore, studying the simplicity of these solution concepts is of fundamental importance. 

A natural measure of simplicity of equilibrium is the size of its support, i.e., the total number of strategy profiles played with nonzero probability. Support size is a pragmatic notion of simplicity, since equilibria with large support size are hard to implement and unlikely to be adopted by human players. This was pointed out early on by Lipton et al.~\cite{LMM} where they proved the existence of small-support approximate\footnote{An $\varepsilon$-approximate equilibrium, where $\varepsilon >0$, is a distribution over strategy profiles at which no player has more than an $\varepsilon$ incentive to deviate.} \emph{Nash} equilibria by generalizing the work of Alth\"{o}fer~\cite{A}.  

In this paper we extend this line of work and establish bounds on the support size of approximate correlated equilibrium and approximate coarse correlated equilibrium in large games; specifically, games with $n$ players and $m$ actions per player. 

The total number of strategy profiles in a game with $n$ players and $m$ actions per player is $m^n$, hence the support size of an equilibrium can be as large as $m^n$. But, both \emph{exact} coarse correlated equilibria and \emph{exact} correlated equilibria are relatively simple solution concepts in terms of their support size. Since correlated equilibria can be specified by $nm(m-1)$ linear inequalities (see Section~\ref{sect:prelim} for details; specifically, consider the definition in Remark \ref{rem:ce} with $\varepsilon=0$ in inequality (\ref{eq:ce-def})) there exists a correlated equilibrium with support of size $O(nm^2)$ (this support-size bound for exact correlated equilibrium appears in~\cite{GL}). Using similar arguments, we can show that there exists a coarse correlated equilibrium with support size $O(nm)$, because they are defined by $nm$ linear inequalities (see Definition~\ref{defn:cce}). Examples \ref{ex:cem2}, \ref{ex:ccem}, and \ref{ex:n} in Appendix~\ref{app:ex} show that these bounds are almost tight.

But what if we are interested in \emph{approximate} correlated equilibrium or \emph{approximate} coarse correlated equilibrium? Can the tight poly$(n,m)$ bounds be significantly improved? In this paper we show that the answer is yes. For both coarse correlated equilibrium (see Theorem \ref{theo:cce}) and correlated equilibrium (see Theorem \ref{theo:ce}) we prove that in any $n$-player $m$-action game, for any fixed $\varepsilon$, there exists an $\varepsilon$-approximate equilibrium with support size poly$(\log m, \log n)$. In fact, the small-support equilibrium, whose existence we establish, is just a uniform distribution (over a multiset of size poly$(\log m, \log n)$), and hence it has a very simple structure. 

Small-support approximate Nash equilibria have been considered in \cite{A}, \cite{LMM}, \cite{HRS}, \cite{FNS}, and \cite{BP}. The best known support-size bound for approximate Nash equilibrium is established in~\cite{BP}, wherein they prove the existence of an approximate Nash equilibrium in which  \emph{each player} uses a mixed strategy (i.e., a distribution over her actions) of support size $O(\log n + \log m)$. Since Nash equilibrium, by definition, is a product of independent distributions (one for each player), this implies the existence of an approximate Nash equilibrium with \emph{overall} support size $O((\log n + \log m)^n)$. For the case of approximate correlated equilibrium in two player games ($n=2$) this bound coincides with our upper bound of Theorem \ref{theo:ce}. In all the other cases, i.e., coarse correlated equilibrium and correlated equilibrium with more than two players, the results developed in this paper exponentially improve upon this bound (see Table~\ref{table:results}).

\begin{table}[h]
\label{table:results}
\begin{center}
\begin{tabular}{| c | c | }
\hline
$\varepsilon$-Approximate Equilibrium & Support Size \\
\hline 
 & \\[-3ex]
Nash & $O\left( \left( \frac{\log n + \log m - \log \varepsilon}{\varepsilon^2} \right)^n \right)$~\cite{BP} \\ 
 & \\[-3ex]
\hline
 & \\[-3ex] 
Correlated & $O\left( \frac{\log m ( \log m + \log n - \log \varepsilon )}{ \varepsilon^4 } \right)$ [current paper] \\ 
 & \\[-3ex]
\hline 
 & \\[-3ex]
Coarse Correlated & $O\left( \frac{\log m + \log n}{ \varepsilon^2} \right)$ [current paper] \\ 
\hline
\end{tabular}
\caption{The existence results. In every $n$-player $m$-action game there exists an $\varepsilon$-approximate equilibrium whose support size is upper bounded by the corresponding entry in the second column.}
\end{center}
\end{table}

Beyond existence, we also consider computational issues related to small-support approximate equilibria. For any fixed $\varepsilon$, we present polynomial-time algorithms for computing $\varepsilon$-approximate coarse correlated equilibrium of support size $O(\log m + \log n)$ and approximate correlated equilibrium of support size $O(m\log m + \log n)$. We also prove that finding an \emph{exact} correlated equilibrium with smallest possible support is \rm{NP}-hard under Cook reductions (see Appendix~\ref{app:hard}).

\section{Preliminaries}
\label{sect:prelim}

We consider \emph{$n$-player $m$-action games}, i.e., games with $n$ players and $m$ actions per player.\footnote{All the results in the paper also generalize to the case where each player has a different number of actions, i.e., player $i$ has $m_i$ actions. For ease of exposition, we assume throughout that all the players have the same number of actions $m$.} We use the following standard notation. The set of players is $[n]=\{1,2,...,n\}$ and the set of actions for any player $i \in [n]$ is $A_i=[m]=\{1,2,...,m\}$. The set of strategy profiles is $A=[m]^n$. Let $(a_i, a_{-i})$ denote a strategy profile in which $a_i$ is the action of the $i$th player and  $a_{-i}$ denotes the strategies chosen by players other than $i$. Players' utilities are normalized between $0$ and $1$; in particular, the payoff function of player $i$ is $u_i: A\rightarrow [0,1]$. The set of probability distributions over a set $B$ is denoted by $\Delta(B)$. The payoff function can be multilinearly extended to $u_i:\Delta(A)\rightarrow [0,1]$. That is, for probability distribution $x \in \Delta(A)$, write $u_i(x)$ to denote the expected payoff of player $i$ under $x$.  

A distribution $x\in \Delta(A)$ is called $k$\emph{-uniform} if it is the uniform distribution over a size-$k$ multiset of strategy profiles. Note that the size of the support of any $k$-uniform distribution is at most $k$.


At a high level, the idea behind the notions of \emph{correlated equilibrium}  (CE) and \emph{coarse correlated equilibrium} (CCE) is the following. Players implement some distribution $x\in \Delta(A)$, which is not necessarily a product distribution. We can interpret such a correlated implementation in terms of a \emph{mediator} that randomizes according to the distribution $x$, i.e., draws an action profile $a=(a_i)_{i\in [n]}$ from $x$. Then the mediator (privately) tells to every player $i$ the corresponding action $a_i$. We will call the drawn action $a_i$ \emph{the recommendation to player $i$}. 

A distribution $x \in \Delta(A)$ is an $\varepsilon$-\emph{coarse correlated equilibrium}\footnote{The set of coarse correlated equilibria is sometimes called the Hannan set, see, e.g.,~\cite{Young}.}  if no player can gain more than $\varepsilon$ by switching to a single pure action $j\in A_i$ instead of following the recommendation of the mediator. 

In addition, we say that a distribution $x \in \Delta(A)$ is an $\varepsilon$-\emph{correlated equilibrium} if no player can gain more than $\varepsilon$ by following any switching rule $f: A_i \rightarrow A_i$ (i.e., by switching from the recommended action $a_i$ to some other action $f(a_i)$).

More formally, we have the following definitions.

\begin{definition}
\label{defn:cce}
Write $R^i_j(a):=u_i(j,a_{-i})-u_i(a)$ to denote the regret of player $i$ for not playing $j$ at strategy profile $a$. A distribution $x\in \Delta(A)$ is an $\varepsilon$-\emph{coarse correlated equilibrium} ($\varepsilon$-CCE) if $\E_{a\sim x} [ R^i_j(a) ] \leq \varepsilon$ for every player $i$ and every action $j\in A_i$.
\end{definition}

\begin{definition}\label{defn:ce}
Write $R^i_f(a):=u_i(f(a_i),a_{-i})-u_i(a)$ to denote the regret of player $i$ for not implementing the switching rule $f$ at strategy profile $a$. A distribution $x\in \Delta(A)$ is an $\varepsilon$-\emph{correlated equilibrium} ($\varepsilon$-CE)  if $\E_{a\sim x} [ R^i_f(a) ] \leq \varepsilon$ for every player $i$ and every mapping $f: A_i \rightarrow A_i$.

\end{definition}

Note that if in the above definitions we set $\varepsilon = 0$, then we obtain the concepts of (exact) coarse correlated and correlated equilibrium. 

\begin{remark}\label{rem:ce}

There is another common definition of $\varepsilon$-correlated equilibrium (see, e.g., \cite{HM}) which requires that no player can gain more than $\varepsilon$ by changing \emph{a single} recommendation, $a_i$, to another action $j$. Formally,
\begin{align}\label{eq:ce-def}
\sum_{a_{-i}\in A_{-i}} \left(u_i(j,a_{-i})-u_i(a_i,a_{-i}) \right)x(a_i,a_{-i})\leq \varepsilon
\end{align}
for every player $i$ and every pair of actions $a_i,j\in A_i$. 

For an exact correlated equilibrium (i.e., with $\varepsilon =0$) these inequalities are satisfied if and only if Definition~\ref{defn:ce} holds (again with $\varepsilon =0$). But, for $\varepsilon >0$, this definition is \emph{not} equivalent to Definition~\ref{defn:ce} of $\varepsilon$-CE. We argue that this definition is vacuous when the number of actions per player, $m$, is large and $\varepsilon$ is a constant. To see this consider, for example, the uniform distribution $x$ over the $k$ actions $\{(j,j,...,j)_{j\in [k]}\}$, where $ 1/\varepsilon \leq k\leq m$. According to this definition $x$ is a $1/k$-correlated equilibrium, \emph{irrespective} of the payoff function. This is because the marginal probability of every action $a_i$  (i.e., $\sum_{a_{-i}} x(a_i, a_{-i})$) is at most $1/k$ and the utilities are between $0$ and $1$.

\end{remark}

\section{Coarse Correlated Equilibrium}

As a warm-up to the correlated equilibrium case, we first prove the existence of small-support $\varepsilon$-CCE.
\subsection{Existence}

\begin{theorem}\label{theo:cce}
Every $n$-player $m$-action game admits a $k$-uniform $\varepsilon$-coarse correlated equilibrium for all 
\begin{equation}
k> \left\lfloor \frac{2(\ln m + \ln n)}{\varepsilon^2} \right\rfloor
\end{equation} 
\end{theorem}

\begin{proof}
The proof is based on the probabilistic method. Let $ \sigma \in \Delta(A)$ be an exact coarse correlated equilibrium of the game; i.e., we have $\E_{a\sim \sigma} [ R^i_j(a) ] \leq 0$ for every player $i$ and every action $j\in A_i$. We sample $k$ action profiles $a(1),a(2),...,a(k) \in A$ independently at random according to the distribution $\sigma$. Denote by $s$ the uniform distribution over $a(1),a(2),...,a(k)$. 

Note that, for any player $i$ and action $j$, $R^i_j(a)$ with $a \sim \sigma$ is a random variable that takes a value between $-1$ and $1$ (since the utilities of players are between $0$ and $1$), and $\E_{a\sim \sigma} [ R^i_j(a) ] \leq 0$. Therefore by Hoeffding's inequality (see \cite{H}) we have 
\begin{align}\label{ineq:cce}
\Pr_{a(1),...,a(k)\sim \sigma} \left( \E_{a\sim s} [R^i_j(a) ] \geq \varepsilon \right)=\Pr \left(\frac{1}{k} \sum_{\ell \in [k]} R^i_j(a(\ell))\geq \varepsilon \right) \leq e^{-\frac{k\varepsilon^2}{2}}.
\end{align}

For player $i$ and action $j$, write $\mathcal{E}_{i,j}$ to denote the event: $\E_{a\sim s} [R^i_j(a) ] \geq \varepsilon $ (equivalently, $\frac{1}{k} \sum_{\ell \in [k]} R^i_j(a(\ell))\geq \varepsilon$). Inequality (\ref{ineq:cce}) implies that for $ k > \left\lfloor \frac{2(\log m + \log n)}{\varepsilon^2} \right\rfloor$ we have $\Pr ( \mathcal{E}_{i,j} ) < \frac{1}{nm}$. Note that there are $nm$ such events, one for every player $i \in [n]$ and action $j \in [m]$. Therefore, via the union bound, we get that with positive probability none of these events will happen; implying that the sampled distribution $s$ is an $\varepsilon$-CCE.
\end{proof}

\subsection{Computation}

The following proposition shows that not only can we prove the existence an $\varepsilon$-CCE with logarithmic support size, but we can efficiently determine it as well.

\begin{proposition} There exists a polynomial (in $n$ and $m$) time randomized algorithm for computing $k$-uniform $\varepsilon$-coarse correlated equilibrium for
\begin{equation}
k> \left\lfloor \frac{2(\ln m + \ln n + \ln 2)}{\varepsilon^2} \right\rfloor
\end{equation}
\end{proposition}

\begin{proof}
If we set $k> \left\lfloor \frac{2(\ln m + \ln n + \ln 2)}{\varepsilon^2} \right\rfloor$ in inequality (\ref{ineq:cce}) then the following holds: $\Pr(\E_{a\sim s} [R^i_j(a)]\geq \varepsilon) <\frac{1}{2nm}$. This implies that the sampled distribution is an $\varepsilon$-CCE with probability at least $1/2$. 

Note that in multi-player games a coarse correlated equilibrium can be efficiently determined. In particular, we can compute a correlated equilibrium in polynomial time using the algorithm\footnote{This algorithm requires the game to be \emph{succinct}---equivalently, a black box that can compute expected utilities under given product distributions (see~\cite{JL} and~\cite{PR} for details). If the game is not succinct (i.e., such a black box does not exist), then as an alternative we can use regret-minimizing dynamics, e.g., regret matching~\cite{hart2005}, to efficiently compute an approximate CE. Starting form an approximate CE, say with approximation guarantee $\varepsilon/2$, instead of an exact CE will worsen the support-size guarantee by at most a constant factor.} of Jiang and Leyton-Brown~\cite{JL} (see also~\cite{PR}). We can then treat the computed correlated equilibrium as a coarse correlated equilibrium.

Now our randomized algorithm is direct, it starts with a coarse correlated equilibrium $\sigma$ and then it samples $k$ pure action profiles according to $\sigma$. Finally, the algorithm checks whether the uniform distribution over the samples forms an $\varepsilon$-CCE or not. If not, then it samples ($k$ action profiles) again. In expectation, after two sampling iterations the algorithm will find an $\varepsilon$-CCE.
\end{proof}

\section{Correlated Equilibrium}
In this section we establish the existence of an $\varepsilon$-CE with polylogarithmic support size. Note that in Definition~\ref{defn:ce}, an $\varepsilon$-CE is specified via $nm^m$ inequalities of the form $\E_{a\sim x} [ R^i_f(a) ]\leq \varepsilon$, one for every $i \in [n]$ and $f: A_i \rightarrow A_i$. Hence simply applying the probabilistic method, as in the case of coarse correlated equilibrium, will not give us the desired polylogarithmic bound. In particular,  sampling from an \emph{arbitrary} correlated equilibrium leads to a support-size guarantee of about $\log (nm^m) = m \log m + \log n$. We get around this issue (see Proof of Theorem~\ref{theo:ce} for details) by sampling from a particular approximate Nash equilibrium for which we only have to consider $nm^{(\log n + \log m)}$ inequalities.

\subsection{Existence}

\begin{theorem}\label{theo:ce}
Every $n$-player $m$-action game admits a $k$-uniform $\varepsilon$-correlated equilibrium for all 
\begin{equation}
k\geq \left\lfloor \frac{264\ln m(\ln m + \ln n - \ln \varepsilon + \ln 16)}{\varepsilon^4} \right\rfloor = O\left( \frac{\log m (\log m + \log n - \log \varepsilon)}{\varepsilon^4} \right)
\end{equation} 
\end{theorem}

\begin{proof}
Let $ \sigma \in\Delta (A)$ be a distribution in which every player $i$ plays actions only from a subset $B_i\subseteq A_i$; i.e., $\sigma(a)>0$ implies $a_i\in B_i$. Then $\sigma$ is an $\varepsilon$-CE iff $\E_{a\sim \sigma}[ R^i_f(a) ]\leq \varepsilon$ for every switching rule $f:B_i\rightarrow A_i$. In other words, we can consider only switching rules $f:B_i\rightarrow A_i$ instead of $f:A_i\rightarrow A_i$, because all the recommendations to player $i$ will be in the set $B_i$. Note that, given a player $i$ and subset $B_i \subseteq A_i$, there are $m^{|B_i|}$ switching rules of the form $f: B_i \rightarrow A_i$. 

In Definition \ref{defn:ce} there are $nm^m$ inequalities ($\E_{a\sim \sigma} [ R^i_f(a) ]\leq \varepsilon$), one for every $i \in [n]$ and mapping $f: A_i \rightarrow A_i$. To avoid dealing with all these $nm^m$ inequalities, we start with an approximate correlated equilibrium $\sigma$ in which every player plays actions from a small subset, i.e., the sets $B_i$ for $\sigma$ are of small cardinality. Then, by the above argument, the number of switching rules (in other words, the number of inequalities of the form $\E_{a\sim \sigma} [ R^i_f(a) ]\leq \varepsilon$) that we need to consider  will be significantly smaller than $nm^m$. Existence of an approximate correlated equilibrium wherein each player uses only a small subset of her actions follows from a result in~\cite{BP}.

Theorem 1 in~\cite{BP} shows that in any $n$-player $m$-action game there exists an $(\varepsilon/2)$-approximate Nash equilibrium $\sigma= \Pi_i \sigma_i$ in which  each player uses a mixed strategy with support size at most $b$ where $b=\left\lceil \frac{32(\ln n + \ln m - ln \varepsilon + \ln 16)}{\varepsilon^2} \right\rceil$. That is, $| \textrm{supp}(\sigma_i) | \leq b$ for all $i$. Since $\sigma$ is an $(\varepsilon/2)$-Nash equilibrium it is an $(\varepsilon/2)$-CE as well. In addition, here, the set $B_i$ is equal to the support of player $i$'s mixed strategy. Therefore, we have $|B_i| \leq b$ for all $i$. 

We now apply the probabilistic method. We sample $k$ action profiles $a(1),a(2),...,a(k) \in A$ independently at random according to the distribution $\sigma$ and denote by $s$ the uniform distribution over the samples. For every player $i$ and a switching rule $f: B_i\rightarrow A_i$, the regret $R^i_f(a)$, with $a \sim \sigma$,  is a random variable that takes a value in $[-1, 1]$ and satisfies: $\E_{a\sim \sigma} [ R^i_f(a) ] \leq \varepsilon/2$. 

Therefore by Hoeffding's inequality (see \cite{H}) we have
\begin{equation}
\Pr \left(\E_{a\sim s} [ R^i_f(a) ] \geq \varepsilon \right)=\Pr \left(\frac{1}{k} \sum_{\ell\in [k]} R^i_f(a( \ell))\geq \varepsilon \right) \leq e^{-\frac{k\varepsilon^2}{8}}.
\end{equation}
Setting $k> \left\lfloor \frac{264\ln m(\ln n + \ln m - \ln \varepsilon + \ln 16)}{\varepsilon^4} \right\rfloor$ guarantees that $\Pr( \E_{a\sim s} [R^i_f(a)]\geq \varepsilon) <\frac{1}{nm^b}$. Since we have at most $nm^b$ such events (one for every player $i\in [n]$ and every switching rule $f: B_i\rightarrow [m]$), union bound implies that with positive probability none of them will happen. Therefore, with positive probability, $s$ is an $\varepsilon$-CE.  
\end{proof}

\subsection{Computation}

Unlike the coarse correlated equilibrium case, there is no guarantee that sampling an arbitrary correlated equilibrium polylogarithmic many times will generate an $\varepsilon$-CE with small support. This is because in the proof of Theorem~\ref{theo:ce} we sampled from a very specific approximate equilibrium; in particular, an approximate correlated equilibrium in which every player uses at most $O\left(\frac{\log n + \log m - \log \varepsilon}{\varepsilon^4} \right)$ strategies. We do not know whether such an approximate correlated equilibrium can be computed in polynomial time.  Nevertheless, we are able to compute $\varepsilon$-CE with support size $O\left(\frac{m\log m +\log n}{\varepsilon^2}\right)$, because $O\left(\frac{m\log m +\log n}{\varepsilon^2}\right)$ samples from \emph{any} correlated equilibrium are enough to form an approximate correlated equilibrium. This algorithm improves upon the know results of \cite{JL} and \cite{HM} that respectively generate an exact and approximate correlated equilibrium with support size $O(nm^2)$. 

\begin{proposition}
There exists a polynomial (in $n$ and $m$) time randomized algorithm for computing $k$-uniform $\varepsilon$-correlated equilibrium for
\begin{equation}
k> \left\lfloor \frac{2(m \ln m + \ln n + \ln 2)}{\varepsilon^2} \right\rfloor.
\end{equation} 
\end{proposition}

\begin{proof}
Let $\sigma$ be a correlated equilibrium. If we sample $k$ pure action profiles according to the distribution $\sigma$ then we have 
\begin{equation}
\Pr(\E_{a\sim s} [R^i_f(a)]\geq \varepsilon)\leq e^{-\frac{k\varepsilon^2}{2}}\leq \frac{1}{2nm^m}.
\end{equation}
for every switching rule $f:A_i\rightarrow A_i$ of every player $i$. Therefore, with probability at least $1/2$ the uniform distribution over the $k$ samples  forms an $\varepsilon$-CE.

Now the algorithm is straightforward. First compute a correlated equilibrium (for example, using the algorithm from~\cite{JL}), then sample $O\left(\frac{m\log m +\log n}{\varepsilon^2}\right)$ actions until the empirical distribution forms an $\varepsilon$-CE.\footnote{We can test in polynomial time whether a distribution with polynomial support size, say $x$, is an $\varepsilon$-CE or not. Specifically, for every player $i$ and each action $a_i \in [m] $ we can first determine $a_i' \in [m]$ that maximizes $\sum_{a_{-i} : x(a_i, a_{-i}) > 0 } (u_i (a_i', a_{-i}) - u_i(a_i, a_{-i}) ) x(a_i, a_{-i})$. Then, set $f(a_i) = a_i'$ and verify that $\E_{a\sim x}[R_f^i(a)] \leq \varepsilon$.}

\end{proof}

The following example demonstrates that if we sample from an arbitrary correlated equilibrium, then in fact we may need $O(m)$ (and not logarithmic) samples to obtain an $\varepsilon$-CE.

\begin{example}
Consider two-player matching-pennies game where in addition to the standard real actions, $r_i\in\{ -1,1 \}$, the two players also choose a dummy number $d_i\in [m]$ that is \emph{irrelevant for the payoffs}. Formally, the action set of each player $i \in [2]$ is $\{(r_i,d_i): r_i\in\{ -1,1 \} \text{ and } d_i \in [m]\}$. The payoffs are given by $u_1((r_i,d_i)_{i=1,2})=r_1 r_2=-u_2((r_i,d_i)_{i=1,2})$. 

Consider the following correlated equilibrium of the game. First we select a $d \in [m]$ uniformly at random, and then set $r_1$, and independently $r_2$,  to be $1$ or $-1$ with equal probability. Note that, for any $d\in [m]$, if we sample from this distribution then the probability that drawn action profile contains $d$---i.e., the drawn action profile is of the form $((r_i,d)_{i=1,2}))$---is equal to $1/m$.

Now, if we sample $m$ strategy profiles from this distribution, then for any $d \in [m]$ the probability that it is picked \emph{exactly} once during the sampling is $m \frac{1}{m} \cdot \left(1-\frac{1}{m} \right)^{m-1}\approx \frac{1}{e}$. If a certain $d$ was picked exactly once then both players can deduce from $d$ which action their opponent will play. Note that the expected number of $d\in [m]$ that are sampled exactly once is $\frac{m}{e}$. Moreover, the probability that the number of exactly-once-sampled $d$'s will be smaller than $\frac{m}{2e}$ is exponentially small in $m$ (see , e.g., Lemma 4 in~\cite{FM}). So, with probability exponentially close to 1, in the resulting uniform distribution at least one player may increase her payoff by at least $\frac{1}{4e}$ by reacting optimally to the opponent's known strategy in all cases in which she got the recommendation $(r_i,d)$, where $d$ was chosen exactly once. Therefore with exponentially high probability the samples does not induce an $\varepsilon$-CE for $\varepsilon <\frac{1}{4e}$. 
\end{example}

\section{Discussion}

Having established polylogarithmic bounds on the support size of approximate correlated equilibrium and approximate coarse correlated equilibrium, it is natural to ask whether these bounds are tight.

Alth\"{o}fer \cite{A} gave an example of a two-players $m$-action zero-sum game where the support of every approximate optimal strategy of one of the players is at least $\Omega(\log m)$. By considering the same game in the context of coarse correlated equilibrium (or correlated equilibrium) we can deduce that for every approximate coarse correlated equilibrium (or approximate correlated equilibrium) in the game the support size is at least $\Omega(\log m)$. Therefore, in the bound of Theorem~\ref{theo:cce} the $\log m$ term is tight. In addition, the $\log^2 m$ term is almost tight in Theorem~\ref{theo:ce}.

On the other hand, the exact dependence of the support size of approximate equilibrium on the number of players, $n$, remains open. To pinpoint this question, we consider $n$-player 2-action games (where the set of approximate correlated equilibria coincides with the set of approximate coarse correlated equilibria) and ask the following questions:

\textbf{Open Question 1:} Does there exist a family of $n$-player 2-action games $\{\Gamma_n \}_{n=1}^\infty$ along with constants $c$ and $\varepsilon$, such that for every $\varepsilon$-CE in $\Gamma_n$ the support size is at least $c\log n$?

A positive answer to this question will mean that the $\log n$ term in Theorems~\ref{theo:cce} and~\ref{theo:ce} is tight.

\textbf{Open Question 2:} Does there exist a function $s(\varepsilon)$ ($s$ does not depend on $n$), such that every $n$-player $2$-action game admits an $\varepsilon$-CE with support size at most $s(\varepsilon)$?

A positive answer to this question will mean that the support size of an approximate equilibrium can be independent of $n$.

\bibliographystyle{plain}
\bibliography{ss-ce}

\appendix

\section{Tightness of the bound by Germano and Lugosi~\cite{GL}}
\label{app:ex}

In a $n$-player $m$-action game with strategy space $A$, correlated equilibria are specified by $nm(m-1)$ linear inequalities in the affine space $\Delta(A)$. Using this fact, Germano and Lugosi~\cite{GL} proved the existence of a correlated equilibrium with support of size $nm(m-1)+1$. Along similar lines we can show that the exists a of coarse correlated equilibrium of support size $nm+1$. 

The following example demonstrates that an $m^2$ term must exist in the support-size bound for exact correlated equilibrium.

\begin{example}\label{ex:cem2}
There exists a two-player $m$-actions game with unique correlated equilibrium where each player randomizes uniformly over all her $m$ actions. This game is a $m$-action generalization of the rock-paper-scissors game, and it is presented in~\cite{N} and~\cite{V}. The support of the unique correlated equilibrium in this game is $m^2$.
\end{example}

The following example demonstrates that a factor of $m$ is unavoidable in the $O(nm)$ bound for exact coarse correlated equilibrium.

\begin{example}\label{ex:ccem}
Consider the following two-player $m$-actions zero sum game where player 1 tries to match player 2, and player 2 tries to miss match.
\begin{equation*}
    u_1(a,b) = \begin{cases}
               1               & \text{if } a=b\\
               0               & \text{otherwise.}
           \end{cases}
\end{equation*}
The payoff of player 2 is defined by $u_2(a,b)=-u_1(a,b)$.
In this game player $1$ can guarantee the value $1/m$ by playing uniformly over all her actions. For every distribution over $A$ with support of size less than $m$ player 2 can get a payoff of $0$ by playing pure strategy. Therefore, a coarse correlated equilibrium of size less than $m$ does not exist.
\end{example}

Next we construct a game in which the support size of any coarse correlated, and hence correlated, equilibrium is at least $n$. This shows that a factor of $n$ is unavoidable in the support-size bound for correlated and coarse correlated equilibrium. We will need the following proposition in order to establish the claim. 

\begin{proposition}\label{lem:sum}
Let $P=\{p_j\}_{j\in [k]}$ be a set of positive reals ($p_j>0$) that can generate all the values in $\{2^{-i}\}_{i\in [n]}$ as partial sums; i.e., there exist subsets of indexes $\{B_i\}_{i \in [n]}$ such that $\sum_{j\in B_i} p_j=2^{-i}$. Then $k\geq n$.
\end{proposition}

\begin{proof}
By multiplying all the elements in $P$ by $2^n$, we have the following equivalent claim that we establish below. Let $P=\{p_j\}_{j\in [k]}$ be a set of positive reals that can generate all the values in $\{2^{i-1}\}_{i\in [n]}$ as partial sums, then $k\geq n$.

We prove by induction that for $i\leq n$ there must be at least $i$ elements in $P$ that are less than or equal to $2^{i-1}$. This statement for $i$ completes the proof.

For $i=1$, we need a partial sum equal to $1$. Therefore, in $P$ ,there must be an element $p \leq 1$.

For $i+1$, by induction we know that there exist $i$ elements, say $p_{1},p_{2},...,p_{i}$, that satisfy $p_{j}\leq 2^{j-1}$. If we sum these elements we have $\sum_{j \in [i]} p_{j} \leq 2^i-1<2^i$. Therefore, in order to have $2^i$ as a partial sum, there must be at least one additional element that satisfies $p_{i+1} <2^i$.
\end{proof}

\begin{example}\label{ex:n}
We construct a $2n$-player $2$-action game in which the support size of any correlated equilibrium is at least $n$.. In $2$-action games the set of correlated equilibria and coarse correlated equilibria coincide. Therefore the example holds for both solution concepts.

\begin{figure}[h!]
\begin{center}
\begin{tabular}{ccc}

\multicolumn{1}{l}{} & 1                          & 2                          \\ \cline{2-3} 
1                    & \multicolumn{1}{|c|}{$v,-v$} & \multicolumn{1}{|c|}{$0,0$}  \\ \cline{2-3} 
2                    & \multicolumn{1}{|c|}{$0,0$}  & \multicolumn{1}{|c|}{$1,-1$} \\ \cline{2-3} 
\end{tabular}
\end{center}
\caption{A $2$-player $2$-action game with unique correlated equilibrium}
\label{table:ex}
\end{figure}

First note the $2$-player $2$-action zero-sum game, with $v>0$, shown in Figure~\ref{table:ex} has a unique correlated equilibrium, which is the Nash equilibrium, wherein both players play the mixed strategy $(\frac{1}{v+1},\frac{v}{v+1})$.

Now consider $n$ pairs of players $(R_i,C_i)_{i\in [n]}$ who play the above game with \emph{different} parameters $v_i$, i.e., for all $i$, we replace $v$ in the above game by $v_i$. Here, the payoffs of players $R_i$ and $C_i$ do not depend on the actions of players $R_j$ and $C_j$ for $j\neq i$. 

Correlated equilibria of this game can be characterized as follows: $x$ is a correlated equilibrium iff, for all $i$, the marginals of the pairs of strategies of players $(R_i,C_i)$ is exactly

\begin{table}[h]
\begin{center}
\begin{tabular}{ccc}

\multicolumn{1}{l}{} & 1                          & 2                          \\ \cline{2-3} 
1                    & \multicolumn{1}{|c|}{$\frac{1}{(v_i+1)^2}$} & \multicolumn{1}{|c|}{$\frac{v_i}{(v_i+1)^2}$}  \\ \cline{2-3} 
2                    & \multicolumn{1}{|c|}{$\frac{v_i}{(v_i+1)^2}$}  & \multicolumn{1}{|c|}{$\frac{v_i^2}{(v_i+1)^2}$} \\ \cline{2-3} 
\end{tabular}
\end{center}
\end{table}
That is, the marginals form the unique correlated equilibrium of the game between players $R_i$ and $C_i$. In particular, the marginals over the strategies of player $R_i$ are $\left(\frac{1}{v_i+1},\frac{v_i}{v_i+1} \right)$.

Note that, by definition, marginals are (partial) sums of probabilities $\sum_{a\in B} x(a)$ with $B\subset \textrm{supp}(x)$. Set $v_i=2^i-1$, then for any correlated equilibrium, $x$, we can take partial sums of the probabilities of strategy profiles in $\textrm{supp}(x)$ to generate all values in $\left\{ \frac{1}{v_i+1} \right\}_{i \in [n]} =\{2^{-i}\}_{i\in [n]}$. Therefore, by Proposition \ref{lem:sum}, $\textrm{supp}(x)$ must contain at least $n$ elements.

\end{example}

\section{Hardness Result}
\label{app:hard}

We prove that finding a correlated equilibrium with smallest possible support is \rm{NP}-hard under Cook reductions, even in two-player zero-sum games. To accomplish this we first show that in a two-player zero-sum game finding a Nash equilibrium with minimum support size is \rm{NP}-Hard. Then, we use a correspondence between correlated equilibria and Nash equilibria in two-player zero-sum games to obtain the result.

A sparsest Nash equilibrium is a Nash equilibrium with minimum support size. In the following lemma we reduce exact cover by $3$ sets, a problem known to be \rm{NP}-hard (see~\cite{gary}), to the problem of finding a sparsest Nash equilibrium.

\begin{lemma}
\label{lem:ne-hard}
Given a two-player zero-sum game, finding a sparsest Nash equilibrium is \rm{NP}-hard under Cook reductions. 
\end{lemma}

\begin{proof}
In the exact cover by $3$ sets  problem (X3C) we are given a universe of elements $J$ and a collection, $\mathcal{I} = \{ S_i \}_{i \in [m]}$, of $3$-element subsets of $J$ and the goal is to determine if there an exact cover of $J$, i.e., a subcollection $\mathcal{I}' \subseteq \mathcal{I}$ such that every element of $J$ is contained in exactly one member of $\mathcal{I}'$. At a high level, given an X3C instance, we construct a two-player zero-sum game in which the first player picks a set in the collection $\mathcal{I}$ and the second player picks an element from $J$. The goal of the first player is to select a set that covers the second player's element and, since it is a zero-sum game, the second player wants to avoid getting covered. Formally, the action sets of the players are $[m]$ and $[n]$ respectively, where $n = |J|$ and the utilities are as follows: $u_1(i,j) = -u_2(i,j) = 1 $ if $j \in S_i$, else if $j \notin S_i$, $u_1(i,j) = -u_2(i,j) = -1$.

Write $(\sigma^*_1, \sigma_2^*)$ to denote a sparsest Nash equilibrium of the game. Here $\sigma_1^*$ and $\sigma_2^*$ are the mixed strategies of the first and second player respectively. Below we prove that $|\textrm{supp}(\sigma_1^*)| = n/3$ iff the given X3C instance has an exact cover. Hence if we are given a sparsest Nash equilibrium we can efficiently determine (by looking at the support size of the mixed strategy of the first player) whether the given X3C instance has an exact cover or not. This completes the reduction. 

We assume that for all $j \in J$ there exists a set $S_i \in \mathcal{I}$ such that $j \in S_i$, else the problem is trivial. Therefore, the value of the game is positive: player one can guarantee a payoff of at least $1/m$ by playing the uniform distribution over $[m]$. Since the value of the game is positive, the second player receives a negative payoff at any Nash equilibrium. Using this we can show that for \emph{any} Nash equilibrium $(\sigma_1, \sigma_2)$ the sets whose index is in the support of $\sigma_1$ cover $J$: $\cup_{ i \in \textrm{supp}(\sigma_1)}  S_i = J $. If this is not the case, i.e., there exists an element $j \in J$ that is not covered by $\cup_{ i \in \textrm{supp}(\sigma_1)}  S_i$, then the second player can play the pure action corresponding to $j$ and get a payoff of $1$, which contradicts the fact that the second player receives a negative payoff at any Nash equilibrium. 

Note that the subsets in $\mathcal{I}$ are of cardinality three, therefore any cover of $J$ must contain at least $n/3$ subsets. This implies that $|\textrm{supp}(\sigma^*_1)| \geq n/3$. This inequality holds regardless of the existence of an exact cover. In particular, if the X3C instance does not have an exact cover then $|\textrm{supp}(\sigma_1^*)| > n/3$.

On the other hand, if the given X3C instance has an exact cover then $|\textrm{supp}(\sigma_1^*)| = n/3$. To show this we first consider the following Nash equilibrium: the mixed strategy of the first player is the uniform distribution over the exact cover and the mixed strategy of the second player is the uniform distribution over $[n]$. Since in two-player zero-sum games mixed strategies of Nash equilibria are interchangeable, we get that $|\textrm{supp}(\sigma_1^*)| = n/3$. Overall, this establishes the desired claim that $|\textrm{supp}(\sigma_1^*)| = n/3$ iff the X3C instance has an exact cover. 
\end{proof}

Note that in a two-player zero-sum game,  $(\sigma_1, \sigma_2)$ is a Nash equilibrium iff the mixed strategies $\sigma_1$ and $\sigma_2$ are optimal strategies\footnote{We work with the standard maxmin definition of optimal strategies. Specifically, $\sigma_i$ is said to be an optimal strategy for player $i$ if it satisfies: $\sigma_i \in \argmax_{ x_i  \in \Delta([m])} \min_{x_{-i} \in \Delta([m])} u_i(x_i, x_{-i})  $.} of the two players. Hence, Lemma~\ref{lem:ne-hard} implies that finding a sparsest optimal strategy, say for the first player, is \rm{NP}-hard.


We now prove the hardness of finding a sparsest (i.e., one with minimum support size) correlated equilibrium.  
\begin{theorem}
\label{thm:ce-hard}
Given a two-player zero-sum game, finding a sparsest correlated equilibrium is \rm{NP}-hard under Cook reductions. 
\end{theorem}

\begin{proof}
Let $\pi$ be a correlated equilibrium of a two-player zero-sum game. It is shown in~\cite{forges} that for any action $a_2$ of the second player such that $\pi(a_2) > 0$ (i.e., $a_2$ is played with positive probability), the conditional probability distribution $\pi \mid a_2$ over the first player's actions is an optimal  strategy for the first player.  We have the same result for actions $a_1$ of the first player (with $\pi(a_1) > 0$)  and conditional probability distribution $\pi \mid a_1$. Therefore, $(\pi \mid a_2, \pi \mid a_1)$ is a Nash equilibrium.

Write $\pi^*$ to denote a sparsest correlated equilibrium of the given game. Also, let $\sigma^* = (\sigma_1^*, \sigma_2^*)$  be a sparsest Nash equilibrium of the game and $r_i = |\textrm{supp}(\sigma_i^*)|$ for $i \in \{1,2\}$. Since $\sigma^*$ is a correlated equilibrium as well, we have 
\begin{align}
\label{ineq:supp}
|\textrm{supp}(\pi^*)| & \leq r_1 r_2. 
\end{align}

In two-player zero-sum games mixed strategies of Nash equilibria are interchangeable; therefore, for any Nash equilibria of the game, $(\sigma_1, \sigma_2)$, we have $ |\textrm{supp}(\sigma_1)| \geq r_1 $ and $ |\textrm{supp}(\sigma_2)| \geq r_2$. In particular, $| \textrm{supp}(\pi^* \mid a_2)| \geq r_1$ and $|\textrm{supp}(\pi^* \mid a_1)| \geq r_2 $ for any two actions $a_1$ and $a_2$ that are played with positive probability.  Therefore, inequality (\ref{ineq:supp}) is tight and we have $| \textrm{supp}(\pi^* \mid a_2)| = r_1$ along with $|\textrm{supp}(\pi^* \mid a_1)| = r_2 $.

The above stated property implies that, given $\pi^*$, we can efficiently determine a sparsest Nash equilibrium.  In particular, let $a_1$ ($a_2$) be an action of the first (second) player such that $\pi^*(a_1) > 0$ ($\pi^*(a_2) > 0$), then ($\pi^* \mid a_2$, $\pi^* \mid a_1$) is a sparsest Nash equilibrium. Overall, using Lemma~\ref{lem:ne-hard} we get the desired result. 
\end{proof}

\end{document}